\documentclass[11pt]{article}

\usepackage[T1]{fontenc}
\usepackage[utf8]{inputenc}
\usepackage{lmodern}
\usepackage{microtype}
\usepackage[a4paper,margin=1in]{geometry}
\usepackage{amsmath,amssymb,mathtools,bm}
\usepackage{amsthm}
\usepackage{booktabs}
\usepackage{array}
\usepackage{float}
\usepackage{xcolor}
\usepackage{graphicx}
\usepackage{enumitem}
\usepackage{lineno}
\usepackage{cite}
\usepackage{hyperref}

\hypersetup{
  colorlinks=true,
  linkcolor=blue!55!black,
  citecolor=blue!55!black,
  urlcolor=blue!55!black,
  pdftitle={Superradiant Scattering and Cavity Instability in a Two-Dimensional Charged String Black Hole},
  pdfauthor={Wen-Xiang Chen}
}

\allowdisplaybreaks
\modulolinenumbers[5]

\newcommand{\dd}{\mathrm{d}}
\newcommand{\ii}{\mathrm{i}}
\newcommand{\ee}{\mathrm{e}}
\newcommand{\cD}{\mathcal{D}}

\newcommand{\cR}{\mathcal{R}}

\newcommand{\Res}{\operatorname*{Res}}
\newcommand{\Var}{\operatorname*{Var}}
\newcommand{\avg}[1]{\left\langle #1\right\rangle}
\newcommand{\abs}[1]{\left|#1\right|}

\title{Superradiant Scattering and Cavity Instability in a Two-Dimensional Charged String Black Hole}

\author{
  Wen-Xiang Chen\thanks{Email: \href{mailto:chenwx@gcu.edu.cn}{chenwx@gcu.edu.cn}}\\[0.35em]
  \small School of Electronic Information, Guangzhou City University of Technology,\\
  \small Guangzhou 510800, China\\
  Yao-Guang Zheng\thanks{Email: \href{mailto:hesoyam12456@163.com}{hesoyam12456@163.com}}\\[0.35em]
  \small Department of Physics, College of Sciences, Northeastern University, Shenyang 110819, China
}

\date{}

\begin{document}
\maketitle

\begin{abstract}
We study the scattering and confinement of charged scalar waves outside a two-dimensional charged string black hole. At fixed dilaton slope $\lambda$, the horizon potential $a$ sets the upper edge $\omega=ea$, and the Hawking temperature determines $a$. For a scalar with the string-frame kinetic weight $h\propto\ee^{-2\Phi_{\rm dil}}$, a wave incident from infinity must also satisfy $\omega>\sqrt{\mu^2+\lambda^2/4}$. Within the nonextremal family, the superradiant band exists precisely when $e>\sqrt{\mu^2+\lambda^2/4}$ and $0<T_H<T_c$, where $T_c=(\lambda/4\pi)[1-(\mu^2+\lambda^2/4)/e^2]$. For zero curvature coupling, the exact reflection coefficient bounds the gain throughout this band by $100\%$. For a massless scalar at $a=0.8$ and $e/\lambda=1$, the peak flux gain falls from $268\%$ for $h=1$ to $26\%$ for the string-frame weight. With a Dirichlet mirror, the fundamental mode first becomes unstable at $\lambda x_m=1.948$ and $3.539$, respectively. We follow this mode and the first two overtones through onset and compare their growth rates with an integral bound.

\end{abstract}


\section{Introduction}

Black-hole superradiance is a scattering process in which the reflected wave
carries more flux than the incident wave.  For rotating black holes, the
extracted quantity is rotational energy; for a static charged black hole, it
is electrostatic energy \cite{Bekenstein1973,Brito2020}.  Repeated
amplification can drive an instability when a mechanism
confines the amplified radiation, as in the black-hole-bomb construction
\cite{PressTeukolsky1972,Cardoso2004,FuruhashiNambu2004,LiZhao2014}.  The
distinction between the following three statements is therefore essential:
\begin{enumerate}[label=(\roman*),leftmargin=2.2em]
  \item an effective potential is negative somewhere;
  \item a real-frequency scattering mode is superradiantly amplified;
  \item a boundary-value problem possesses a mode with
        $\operatorname{Im}\omega>0$.
\end{enumerate}
None of these statements is, by itself, equivalent to either of the other
two.

Two-dimensional dilaton-gravity models provide controlled laboratories for
black-hole physics \cite{Grumiller2002}.  Familiar examples include the
$SL(2,\mathbb R)/U(1)$ string black hole \cite{Witten1991}, the CGHS model
\cite{CGHS1992}, and charged string backgrounds
\cite{McGuigan1992,Giveon2003,GiveonKutasov2006}.  Scalar perturbations and
quasinormal spectra have also been studied in other $(1+1)$-dimensional
geometries \cite{Cruz2016}.

In the charged two-dimensional background, earlier work studied tachyon
propagation near the inner horizon \cite{LeeMyungKim1995}, exact charged
Dirac quasinormal frequencies and greybody factors
\cite{Becar2014,Becar2015}, and Hawking emission from Bogoliubov
coefficients \cite{Christodoulakis2001}.  The coset construction gives scattering states on the
extended geometry \cite{Giveon2003}. Its reflection amplitude also enters
the critical black-brane calculation of Ref.~\cite{Gursoy2022}.  Here we
consider a charged scalar incident from the exterior, or confined outside
the future horizon by a reflecting wall.

There are two recurrent sources of confusion in a strictly
$(1+1)$-dimensional perturbation problem.  First, there is no two-sphere and
therefore no independent angular quantum number $\ell$.  An
$\ell(\ell+1)/r^2$ term is meaningful only when a higher-dimensional field is
decomposed into spherical harmonics before dimensional reduction.  Second, a
real negative potential can support a tachyonic or bound-state instability,
but a self-adjoint scattering problem with only positive-energy channels
cannot produce $\abs{\mathcal R}>1$.  Superradiance requires a horizon channel
whose conserved energy or charge flux has the opposite sign.

The charged superradiant threshold, its Wronskian derivation and the absence
of amplification for a neutral probe in a static background are established
results \cite{Bekenstein1973,Brito2020,Manogue1988,Richartz2009}.  We also
use the standard integrated radial equation to constrain growing modes.
For Reissner--Nordstr\"om black holes, earlier work distinguished the
asymptotically flat problem \cite{HodNoBomb2013} from the mirror problem,
for which growth rates, instability thresholds, hairy solutions and
nonlinear evolution have been studied
\cite{Degollado2013,Degollado2014,HodCharge2016,Dolan2015,SanchisGual2016}.
Li and Zhao \cite{LiZhao2014} used asymptotic matching for a massless
charged scalar confined around a four-dimensional string black hole.
Here the geometry is intrinsically two-dimensional.  We compare two
scalar kinetic weights on the same charged background.  For the
string-frame weight, the propagation threshold at infinity and the electrostatic
edge at the horizon cease to overlap once the black-hole temperature reaches
a critical value at fixed probe charge, mass and dilaton slope, provided $e>\sqrt{\mu^2+\lambda^2/4}$.  We calculate how the weight changes the
amplification and the complex frequencies of modes confined by a mirror.

Our analysis is a fixed-background, linear probe calculation.  It does not
assert that every superradiant scattering frequency becomes unstable, nor
does it determine the nonlinear endpoint of a possible instability.  These
scope restrictions are important for keeping the conclusions directly tied
to the derived equations.

\section{Assumptions, conventions, and scalar action}
\label{sec:setup}

We use units $c=\hbar=k_B=1$ and metric signature $(-,+)$.  Consider the
static exterior geometry
\begin{equation}
  \dd s^2=-F(x)\,\dd t^2+\frac{\dd x^2}{F(x)},
  \label{eq:metric}
\end{equation}
where a nonextremal event horizon is located at $x=x_H$:
\begin{equation}
  F(x_H)=0,\qquad
  F'(x_H)=2\kappa>0,\qquad
  F(x)>0\quad (x>x_H).
  \label{eq:horizon}
\end{equation}
Here $\kappa$ is the surface gravity.  The background gauge field and a
positive scalar weight are
\begin{equation}
  A=A_t(x)\,\dd t,\qquad h(x)>0.
  \label{eq:gauge-weight}
\end{equation}
The weight $h$ may be constant in an intrinsic two-dimensional theory or may
encode a dilaton/area factor inherited from dimensional reduction.

Let $\Psi$ be a complex scalar of charge $e$, mass $\mu$, and curvature
coupling $\xi$.  Define
\begin{equation}
  \cD_\mu=\nabla_\mu+\ii e A_\mu,
  \qquad
  U(x)=\mu^2+\xi\cR(x),
  \label{eq:covder}
\end{equation}
where $\cR$ is the two-dimensional Ricci scalar.  The probe action is
\begin{equation}
  S_\Psi
  =-\int\dd^2x\,\sqrt{-g}\,h(x)
  \left[
    g^{\mu\nu}(\cD_\mu\Psi)^*\cD_\nu\Psi
    +U(x)\abs{\Psi}^2
  \right].
  \label{eq:action}
\end{equation}
The sign convention in \eqref{eq:covder} is chosen so that the separated
frequency appears as $\omega-eA_t$.  Reversing the signs of both $e$ and $A_\mu$ leaves their product unchanged.

Varying \eqref{eq:action} with respect to $\Psi^*$ and integrating the
kinetic term by parts gives
\begin{align}
  \delta_{\Psi^*}S_\Psi
  &=-\int\dd^2x\,\sqrt{-g}\,h
  \left[
    g^{\mu\nu}(\cD_\mu\delta\Psi)^*\cD_\nu\Psi
    +U\,\delta\Psi^*\Psi
  \right]
  \nonumber\\
  &=\int\dd^2x\,\delta\Psi^*
  \left[
    \cD_\mu\!\left(\sqrt{-g}\,h\,g^{\mu\nu}\cD_\nu\Psi\right)
    -\sqrt{-g}\,h\,U\Psi
  \right],
  \label{eq:variation}
\end{align}
after discarding the boundary variation.  Hence
\begin{equation}
  \frac{1}{\sqrt{-g}\,h}
  \cD_\mu\!\left(\sqrt{-g}\,h\,g^{\mu\nu}\cD_\nu\Psi\right)
  -U\Psi=0.
  \label{eq:weightedKG}
\end{equation}
Equation \eqref{eq:weightedKG} gives the radial Schr\"odinger equation after the field redefinition below.

\section{Exact reduction to a one-dimensional wave equation}
\label{sec:reduction}

For the metric \eqref{eq:metric},
\begin{equation}
  g_{\mu\nu}=
  \begin{pmatrix}
    -F&0\\[0.2em]0&F^{-1}
  \end{pmatrix},
  \qquad
  g^{\mu\nu}=
  \begin{pmatrix}
    -F^{-1}&0\\[0.2em]0&F
  \end{pmatrix},
  \qquad
  \sqrt{-g}=1.
  \label{eq:metriccomponents}
\end{equation}
Separate the time dependence as
\begin{equation}
  \Psi(t,x)=\ee^{-\ii\omega t}R(x).
  \label{eq:separation}
\end{equation}
Because
\begin{equation}
  \cD_t\Psi
  =(\partial_t+\ii eA_t)\Psi
  =-\ii\bigl(\omega-eA_t\bigr)\Psi,
  \label{eq:Dt}
\end{equation}
substitution of \eqref{eq:metriccomponents}--\eqref{eq:Dt} into
\eqref{eq:weightedKG} yields
\begin{equation}
  \frac{1}{h}\frac{\dd}{\dd x}
  \left(hF\frac{\dd R}{\dd x}\right)
  +\left[
    \frac{(\omega-eA_t)^2}{F}-U
  \right]R=0.
  \label{eq:radialx}
\end{equation}

Introduce the tortoise coordinate
\begin{equation}
  \frac{\dd r_*}{\dd x}=\frac{1}{F(x)}.
  \label{eq:tortoise}
\end{equation}
We denote $\dd/\dd r_*$ by an overdot only in the present derivation.  Since
\begin{equation}
  \frac{\dd R}{\dd x}=\frac{\dot R}{F},
  \qquad
  hF\frac{\dd R}{\dd x}=h\dot R,
  \qquad
  \frac{\dd}{\dd x}(h\dot R)=\frac{1}{F}\frac{\dd}{\dd r_*}(h\dot R),
  \label{eq:tortoisesteps}
\end{equation}
multiplication of \eqref{eq:radialx} by $F$ gives
\begin{equation}
  \ddot R+\frac{\dot h}{h}\dot R
  +\left[(\omega-eA_t)^2-FU\right]R=0.
  \label{eq:Rfirstderivative}
\end{equation}

The first derivative is removed by the local rescaling (unique up to a nonzero constant)
\begin{equation}
  \psi=\sqrt{h}\,R,\qquad R=h^{-1/2}\psi.
  \label{eq:rescale}
\end{equation}
Writing $H=\ln h$, direct differentiation gives
\begin{align}
  \dot R
  &=h^{-1/2}\left(\dot\psi-\frac{\dot H}{2}\psi\right),
  \label{eq:firstderivative}\\
  \ddot R
  &=h^{-1/2}\left[
    \ddot\psi-\dot H\dot\psi
    +\left(\frac{\dot H^2}{4}-\frac{\ddot H}{2}\right)\psi
  \right].
  \label{eq:secondderivative}
\end{align}
Adding $(\dot h/h)\dot R=\dot H\dot R$ to
\eqref{eq:secondderivative}, the $\dot\psi$ terms cancel:
\begin{equation}
  \ddot R+\frac{\dot h}{h}\dot R
  =h^{-1/2}
  \left[
    \ddot\psi
    -\left(\frac{\ddot H}{2}+\frac{\dot H^2}{4}\right)\psi
  \right].
  \label{eq:cancellation}
\end{equation}
Therefore the exact radial equation is
\begin{equation}
  \frac{\dd^2\psi}{\dd r_*^2}
  +\left[
    (\omega-eA_t)^2-V_{\rm eff}
  \right]\psi=0,
  \label{eq:schrodinger}
\end{equation}
with
\begin{equation}
  \boxed{
  V_{\rm eff}
  =F(\mu^2+\xi\cR)
  +\frac{1}{2}\frac{\dd^2\ln h}{\dd r_*^2}
  +\frac{1}{4}
   \left(\frac{\dd\ln h}{\dd r_*}\right)^2
  }.
  \label{eq:Veff}
\end{equation}

For an intrinsic, minimally coupled, neutral massless scalar,
\begin{equation}
  e=0,\qquad \mu=0,\qquad \xi=0,\qquad h=\text{constant}.
  \label{eq:strict2d}
\end{equation}
Every term in \eqref{eq:Veff} vanishes, and
\begin{equation}
  \frac{\dd^2\psi}{\dd r_*^2}+\omega^2\psi=0.
  \label{eq:free2d}
\end{equation}
This is the expected consequence of conformal flatness and the conformal
invariance of a massless scalar in two dimensions.  In particular, an
angular term $\ell(\ell+1)F/r^2$ cannot be appended to
\eqref{eq:free2d}.  Such a term belongs to a higher-dimensional harmonic
reduction; its origin is shown explicitly in Appendix~\ref{app:reduction}.

\section{Flux identity and the superradiant condition}
\label{sec:flux}

Assume that the incident frequency $\omega$ is positive and that $A_t$ and $V_{\rm eff}$ are real.
Taking the complex conjugate of \eqref{eq:schrodinger}, multiplying the
original equation by $\psi^*$, multiplying the conjugate equation by
$\psi$, and subtracting, we obtain
\begin{equation}
  \frac{\dd}{\dd r_*}
  \left(
    \psi^*\frac{\dd\psi}{\dd r_*}
    -\psi\frac{\dd\psi^*}{\dd r_*}
  \right)=0.
  \label{eq:Wderivative}
\end{equation}
Thus the Wronskian
\begin{equation}
  W[\psi,\psi^*]
  =\psi^*\psi'-\psi\psi'^*,
  \qquad {}'=\frac{\dd}{\dd r_*},
  \label{eq:Wronskian}
\end{equation}
is constant.

Fix the asymptotic gauge by
\begin{equation}
  A_t(\infty)=0,\qquad
  \Phi_H=A_t(x_H)-A_t(\infty)=A_t(x_H).
  \label{eq:asymptoticgauge}
\end{equation}
Suppose also that $V_{\rm eff}\to V_\infty$ at spatial infinity and that
\begin{equation}
  k=\sqrt{\omega^2-V_\infty}>0.
  \label{eq:kdef}
\end{equation}
We assume the plane-wave asymptotics
\begin{align}
  \psi
  &\sim
  \mathcal I\,\ee^{-\ii kr_*}
  +\mathcal R\,\ee^{+\ii kr_*},
  &&r_*\to+\infty,
  \label{eq:infinityBC}\\
  \psi
  &\sim
  \mathcal T\,\ee^{-\ii(\omega-e\Phi_H)r_*},
  &&r_*\to-\infty.
  \label{eq:horizonBC}
\end{align}
The sign in \eqref{eq:horizonBC} is fixed by regularity on the future
horizon.  Evaluating \eqref{eq:Wronskian} at infinity gives
\begin{equation}
  W_\infty
  =2\ii k\left(\abs{\mathcal R}^2-\abs{\mathcal I}^2\right),
  \label{eq:Winfinity}
\end{equation}
where the oscillatory cross terms cancel.  At the horizon,
\begin{equation}
  W_H
  =-2\ii(\omega-e\Phi_H)\abs{\mathcal T}^2.
  \label{eq:Whorizon}
\end{equation}
Equating \eqref{eq:Winfinity} and \eqref{eq:Whorizon}, and choosing unit
incident amplitude $\mathcal I=1$, yields
\begin{equation}
  \boxed{
  \abs{\mathcal R}^2
  =1-\frac{\omega-e\Phi_H}{k}\abs{\mathcal T}^2
  }.
  \label{eq:fluxidentity}
\end{equation}

For the scattering problem \eqref{eq:schrodinger}--\eqref{eq:horizonBC},
with $k>0$ and a nonzero transmitted wave, the standard flux criterion
\cite{Bekenstein1973,Brito2020,Richartz2009} gives amplification precisely when
\begin{equation}
  \omega-e\Phi_H<0.
  \label{eq:covfreqnegative}
\end{equation}
For $\omega>0$ and $e\Phi_H>0$, this becomes
\begin{equation}
  0<\omega<e\Phi_H.
  \label{eq:superradiantwindow}
\end{equation}
Charge conjugation sends $(\omega,e)$ to $(-\omega,-e)$ for real $\omega$.
If negative frequencies are retained, the signed asymptotic wave number for
the outgoing wave is $\operatorname{sgn}(\omega)k$, so its group velocity is
positive. With this convention the amplification criterion for either sign is
$\omega(\omega-e\Phi_H)<0$, provided $k>0$ and $\mathcal T\ne0$.
At $k=0$ the incident flux vanishes and a flux gain is undefined.
For $e=0$ and $\omega>0$, a static background of the form
\eqref{eq:metric} cannot superradiantly amplify a scalar wave:
\begin{equation}
  \abs{\mathcal R}^2
  =1-\frac{\omega}{k}\abs{\mathcal T}^2
  \leq1.
  \label{eq:neutralnoSR}
\end{equation}

The derivation did not assume that $V_{\rm eff}$ is positive.  A negative
well changes $\mathcal T$ and can support bound states, but it does not alter
the sign relation in \eqref{eq:fluxidentity}.  Superradiance is controlled
by the sign of the horizon channel, not by the sign of the local potential.

\section{Near-horizon Frobenius index}
\label{sec:residue}

Let $z=x-x_H$ and assume that the background coefficients are analytic at
a simple horizon, with $h(x_H)>0$. Since $F=2\kappa z+O(z^2)$,
$r_*=(2\kappa)^{-1}\ln(z/z_0)+O(z)$. For real $\omega\ne e\Phi_H$ the
ingoing solution has the form $\psi_{\rm in}=z^{-{\ii}(\omega-e\Phi_H)/(2\kappa)}[1+O(z)]$.
Its Frobenius index is the residue of the logarithmic differential,
\begin{equation}
  \boxed{
  \Res_{z=0}\bigl(\dd\ln\psi_{\rm in}\bigr)
  =-\frac{\ii(\omega-e\Phi_H)}{2\kappa}
  =-\frac{\ii(\omega-e\Phi_H)}{4\pi T_H}
  },
  \label{eq:residue}
\end{equation}
where $T_H=\kappa/(2\pi)$. The residue is invariant under analytic local
changes $\widetilde z=cz+O(z^2)$ with $c\ne0$. The logarithmic tortoise
coordinate lies outside this class. At $\omega=e\Phi_H$ the indicial
equation has a double zero root. The regular solution tends to a constant,
and a second solution contains $\psi_{\rm reg}\ln z$. Continuation of the
ingoing solution selects the regular branch, which carries zero horizon
flux. Equation~\eqref{eq:residue} thus expresses the same threshold as
the flux identity.

\section{Application to the charged two-dimensional string black hole}
\label{sec:stringBH}

We now apply the general result to the charged two-dimensional black hole
obtained in string theory and described by an exact
$[SL(2,\mathbb R)\times U(1)]/U(1)$ coset
\cite{McGuigan1992,Giveon2003,GiveonKutasov2006}.  In a
Schwarzschild-like radial coordinate $r$, the Lorentzian fields can be
written as
\begin{align}
  \dd s^2
  &=-f(r)\dd t^2
  +\frac{\dd r^2}{\lambda^2r^2f(r)},
  \label{eq:stringmetric}\\
  f(r)
  &=1-\frac{2m}{r}+\frac{q_b^2}{r^2}
   =\left(1-\frac{r_+}{r}\right)
    \left(1-\frac{r_-}{r}\right),
  \label{eq:stringf}\\
  A_t(r)&=\frac{q_b}{r},
  \qquad
  \Phi_{\rm dil}(r)
  =\Phi_0-\frac{1}{2}\ln\left(\frac{r}{r_0}\right).
  \label{eq:stringfields}
\end{align}
Here $\lambda>0$ is the linear-dilaton slope, while $q_b$ denotes the
black-hole charge parameter and must not be confused with the probe charge
$e$; $\Phi_0$ and $r_0$ fix the irrelevant additive normalization of the
dilaton.  The two horizons are
\begin{equation}
  r_\pm=m\pm\sqrt{m^2-q_b^2},
  \qquad
  r_+r_-=q_b^2.
  \label{eq:horizons}
\end{equation}
We take $m>\abs{q_b}$ in the nonextremal calculation and choose $q_b>0$ for
definiteness.

\subsection{Canonical radial coordinate}

Define
\begin{equation}
  x=\frac{1}{\lambda}\ln\left(\frac{r}{r_+}\right),
  \qquad
  r=r_+\ee^{\lambda x},
  \qquad
  \dd r=\lambda r\,\dd x.
  \label{eq:xtransform}
\end{equation}
Then
\begin{equation}
  \frac{\dd r^2}{\lambda^2r^2f(r)}
  =\frac{\dd x^2}{F(x)}
  \label{eq:radialtransform}
\end{equation}
and the metric takes the canonical form \eqref{eq:metric}, with the outer
horizon at $x=0$.  Introduce the dimensionless charge-to-mass variable
\begin{equation}
  a=\frac{q_b}{r_+}
  =\sqrt{\frac{r_-}{r_+}},
  \qquad 0\leq a<1.
  \label{eq:adef}
\end{equation}
Equations \eqref{eq:stringf} and \eqref{eq:stringfields} become
\begin{align}
  F(x)
  &=\left(1-\ee^{-\lambda x}\right)
    \left(1-a^2\ee^{-\lambda x}\right),
  \label{eq:Fx}\\
  A_t(x)&=a\ee^{-\lambda x}.
  \label{eq:Atx}
\end{align}
This gauge has $A_t(\infty)=0$, so
\begin{equation}
  \Phi_H=A_t(0)-A_t(\infty)=a.
  \label{eq:PhiHa}
\end{equation}
Only the potential difference is physical; a constant shift of $A_t$ must
be accompanied by the corresponding frequency shift.

In the two-derivative string-frame description, the background action in
our normalization is
\begin{equation}
 S_{\rm bg}=\frac{1}{2\kappa_2^2}\int\dd^2x\,\sqrt{-g}\,\ee^{-2\phi}
 \left[\cR+4(\nabla\phi)^2+\lambda^2
             -\tfrac12\mathcal F_{\mu\nu}\mathcal F^{\mu\nu}\right],
 \qquad \phi=\Phi_{\rm dil},\quad \mathcal F=\dd A ,
 \label{eq:bgaction}
\end{equation}
where $\kappa_2$ is the gravitational coupling. The field equations, with
the dilaton equation used to simplify the metric equation, are
\begin{align}
 \cR_{\mu\nu}+2\nabla_\mu\nabla_\nu\phi
      -\mathcal F_{\mu\rho}\mathcal F_{\nu}{}^{\rho}&=0,\nonumber\\
 \nabla_\mu(\ee^{-2\phi}\mathcal F^{\mu\nu})&=0,\nonumber\\
 \cR+4\Box\phi-4(\nabla\phi)^2+\lambda^2
                   -\tfrac12\mathcal F_{\mu\nu}\mathcal F^{\mu\nu}&=0.
 \label{eq:bgequations}
\end{align}
To verify the solution, write $y=\ee^{-\lambda x}$ and
$\phi=\phi_H-\lambda x/2$. Then $\ee^{-2\phi}A_t'=-\lambda a\ee^{-2\phi_H}$
is constant, as required by the Maxwell equation. The metric and dilaton
equations reduce, respectively, to
\begin{align}
 F''+\lambda F' &= 2(A_t')^2,\nonumber\\
 F''+2\lambda F'+\lambda^2(F-1) &= (A_t')^2.
 \label{eq:bgcheck}
\end{align}
With $F=1-(1+a^2)y+a^2y^2$, the left-hand sides are
$2\lambda^2 a^2y^2$ and $\lambda^2a^2y^2$, respectively.

Our $A_t=q_b/r$ and $\lambda$ are the potential and slope $Q$ of
Eqs.~(2.1)--(2.4) of Ref.~\cite{GiveonKutasov2006}, with the time
coordinate normalized at infinity. In Eqs.~(33)--(35) of
Ref.~\cite{Giveon2003}, rescaling time to this normalization gives the
horizon-regular gauge $A_t=q_b/r-a$ (their $p=a$). Both choices give
$\Phi_H=a$. The heterotic action in Eq.~(6.3) of Ref.~\cite{McGuigan1992}
uses $-\mathcal F_{\rm M}^2/4$, with $\mathcal F_{\rm M}=\sqrt{2}\,\mathcal F$
and central-charge parameter $c=-\lambda^2$. In the same gauge the probe
charge in that convention is $e_{\rm M}=e/\sqrt{2}$, so $e_{\rm M}A_{\rm M}=eA$.
The integration constant $q_b$ labels the background electric flux, whereas
$e$ is the independent coupling in the probe action. With $t$ and $x$
measured in length units, $A_t$ and $a$ are dimensionless, and
$\lambda$, $e$, $\omega$, $\mu$ and $T_H$ have dimension $L^{-1}$.
\subsection{Horizon potential and Hawking temperature}

Differentiate \eqref{eq:Fx}:
\begin{align}
  \frac{\dd F}{\dd x}
  &=
  \lambda\ee^{-\lambda x}
  \left[
    1+a^2-2a^2\ee^{-\lambda x}
  \right],
  \label{eq:Fprime}\\
  F'(0)&=\lambda(1-a^2).
  \label{eq:FprimeH}
\end{align}
It follows that
\begin{equation}
  \kappa=\frac{F'(0)}{2}
  =\frac{\lambda}{2}(1-a^2),
  \qquad
  T_H=\frac{\lambda}{4\pi}(1-a^2).
  \label{eq:temperature}
\end{equation}
Solving the second relation for $a$ gives
\begin{equation}
  a=\sqrt{1-\frac{4\pi T_H}{\lambda}}.
  \label{eq:aofT}
\end{equation}
Combining \eqref{eq:superradiantwindow}, \eqref{eq:PhiHa}, and
\eqref{eq:aofT}, we obtain
\begin{equation}
  \boxed{
  0<\omega
  <e\sqrt{1-\frac{4\pi T_H}{\lambda}}
  }
  \qquad
  (e>0,\ 0<T_H\leq\lambda/4\pi).
  \label{eq:temperaturewindow}
\end{equation}
For the intrinsic massless scalar, $k=\omega$ and \eqref{eq:temperaturewindow}
gives the full interval. At fixed $\lambda$, $T_H$ determines $a$ in the linear probe equation. The electric potential fixes the upper edge. This
edge approaches $e$ as $a\to1^-$, while the surface gravity tends to zero.
The limiting amplification is discussed in Section~\ref{sec:numerics}.

\subsection{Dilaton weight and the propagation threshold}
\label{sec:dilatongap}
The weight $h=1$ describes an intrinsic scalar with a canonical kinetic
term. A scalar whose kinetic term inherits the tree-level string-frame
measure has $h\propto\ee^{-2\Phi_{\rm dil}}$; the neutral tachyon kinetic
term in Ref.~\cite{LeeMyungKim1995} is one example of this measure.
Here the probe charge is independent and $\mu^2\geq0$. The tachyon
potential of that string model does not enter our probe action. We set
$h(0)=1$, so that $h=\ee^{\lambda x}$. From \eqref{eq:Veff},
\begin{equation}
 V_{\rm eff}=F(\mu^2+\xi\cR)+\tfrac{\lambda}{2}FF'
                +\tfrac{\lambda^2}{4}F^2,\qquad
 V_\infty=\mu^2+\tfrac{\lambda^2}{4}\equiv M_\infty^2 .
 \label{eq:dilatongap}
\end{equation}
Primes here denote $x$ derivatives. For \eqref{eq:Fx},
$\cR=-F''=\lambda^2[(1+a^2)\ee^{-\lambda x}-4a^2\ee^{-2\lambda x}]$.
Hence $\xi F\cR$ changes the potential in the exterior and its
near-horizon slope, but vanishes at both ends. The following channel
thresholds also apply at finite $\xi$. We set $\xi=0$ in the numerical
examples.

At infinity a positive-frequency incident wave requires $\omega>M_\infty$.
In the gauge $A_t(\infty)=0$, the inward flux of conserved time-translation energy at the horizon is negative for $0<\omega<ea$ ($e>0$).
The two requirements give
\begin{align}
 M_\infty &<\omega< e\sqrt{1-4\pi T_H/\lambda}, \label{eq:gapwindow}\\
 0<T_H&<T_c,\qquad
 T_c=\frac{\lambda}{4\pi}\left(1-\frac{M_\infty^2}{e^2}\right),
 \qquad e>M_\infty. \label{eq:Tcritical}
\end{align}
For $\mu=0$, the lower edge is $\lambda/2$ and a band requires
$e>\lambda/2$. At $T_H=T_c$ its width is zero. If $e\leq M_\infty$,
no nonextremal member of the family has a propagating superradiant band.
The threshold concerns scattering from infinity. A finite mirror imposes
a different outer boundary condition.

\section{Amplification is not instability: an exact cavity identity}
\label{sec:instability}

To determine whether amplification can become an instability, impose a
Dirichlet wall at a finite exterior point $x=x_m$:
\begin{equation}
  \psi(r_*^m)=0,
  \qquad
  r_*^m=r_*(x_m).
  \label{eq:mirrorBC}
\end{equation}
Let
\begin{equation}
  \omega=\omega_R+\ii\omega_I,
  \qquad \omega_I>0,
  \label{eq:complexomega}
\end{equation}
represent a putative exponentially growing mode.  At the future horizon,
\begin{equation}
  \psi\sim
  \ee^{-\ii(\omega-e\Phi_H)r_*}
  =
  \ee^{-\ii(\omega_R-e\Phi_H)r_*}
  \ee^{\omega_I r_*}.
  \label{eq:growinghorizon}
\end{equation}
Because $r_*\to-\infty$, both $\psi$ and $\psi'$ vanish at the lower
integration limit.
To see horizon regularity, introduce $v=t+r_*$ and transform to the
regular potential $A=A_t\,\dd v$ with $\dd\chi=A_t\,\dd r_*$.
In the convention \eqref{eq:covder}, $\Psi\mapsto \ee^{-\ii e\chi}\Psi$.
Since $\chi-\Phi_H r_*$ is regular, the ingoing field is $\ee^{-\ii\omega v}$
times a regular radial factor in this gauge. We integrate on truncated
constant-$t$ exterior slices and take their inner endpoint to $r_*=-\infty$.
The factor $\ee^{2\omega_I r_*}$ makes the radial norm in
\eqref{eq:weightedaverage} finite. This auxiliary norm is used only in
the integrated radial equation.

Multiply \eqref{eq:schrodinger} by $\psi^*$ and integrate from the horizon
to the mirror:
\begin{equation}
  \int_{-\infty}^{r_*^m}
  \psi^*\psi''\,\dd r_*
  +\int_{-\infty}^{r_*^m}
  \left[(\omega-eA_t)^2-V_{\rm eff}\right]
  \abs{\psi}^2\,\dd r_*=0.
  \label{eq:integralstart}
\end{equation}
Integration by parts gives
\begin{equation}
  \int\psi^*\psi''\,\dd r_*
  =
  \left[\psi^*\psi'\right]_{-\infty}^{r_*^m}
  -\int\abs{\psi'}^2\,\dd r_*.
  \label{eq:parts}
\end{equation}
The boundary term vanishes by \eqref{eq:mirrorBC} and
\eqref{eq:growinghorizon}.  Therefore
\begin{equation}
  \int_{-\infty}^{r_*^m}
  \left(\abs{\psi'}^2+V_{\rm eff}\abs{\psi}^2\right)\dd r_*
  =
  \int_{-\infty}^{r_*^m}
  (\omega-eA_t)^2\abs{\psi}^2\,\dd r_*.
  \label{eq:masterintegral}
\end{equation}
Since
\begin{equation}
  (\omega-eA_t)^2
  =
  (\omega_R-eA_t)^2-\omega_I^2
  +2\ii\omega_I(\omega_R-eA_t),
  \label{eq:complexsquare}
\end{equation}
the imaginary part of \eqref{eq:masterintegral} is
\begin{equation}
  2\omega_I
  \int_{-\infty}^{r_*^m}
  (\omega_R-eA_t)\abs{\psi}^2\,\dd r_*=0.
  \label{eq:imaginarypart}
\end{equation}
Define
\begin{equation}
  N=\int_{-\infty}^{r_*^m}\abs{\psi}^2\dd r_*,
  \qquad
  \avg{X}
  =\frac{1}{N}
   \int_{-\infty}^{r_*^m}X\abs{\psi}^2\dd r_*.
  \label{eq:weightedaverage}
\end{equation}
Because $\omega_I>0$, \eqref{eq:imaginarypart} gives the real part of the growing-mode frequency,
\begin{equation}
  \boxed{\omega_R=e\avg{A_t}}.
  \label{eq:omegarealaverage}
\end{equation}
Konoplya and Zhidenko use the same integration by parts and take the imaginary part with an outgoing condition at the cosmological horizon in Reissner--Nordstr\"om--de Sitter spacetime (Eqs.~(12)--(13) of Ref.~\cite{KonoplyaZhidenko2014}).
This relation is gauge covariant: under a constant shift
$A_t\mapsto A_t+C$, the separated frequency changes as
$\omega_R\mapsto\omega_R+eC$.  The variance and all potential differences
used below are gauge invariant.

The real part of \eqref{eq:masterintegral} gives
\begin{equation}
  \int
  \left(\abs{\psi'}^2+V_{\rm eff}\abs{\psi}^2\right)\dd r_*
  =
  \int(\omega_R-eA_t)^2\abs{\psi}^2\dd r_*
  -\omega_I^2N.
  \label{eq:realpart}
\end{equation}
Using \eqref{eq:omegarealaverage},
\begin{align}
  \frac{1}{N}
  \int(\omega_R-eA_t)^2\abs{\psi}^2\dd r_*
  &=
  e^2\avg{\left(A_t-\avg{A_t}\right)^2}
  \nonumber\\
  &=e^2\Var(A_t).
  \label{eq:variance}
\end{align}
Consequently,
\begin{equation}
  \boxed{
  \omega_I^2
  =
  e^2\Var(A_t)
  -
  \frac{
    \displaystyle\int_{-\infty}^{r_*^m}
    \left(\abs{\psi'}^2+V_{\rm eff}\abs{\psi}^2\right)\dd r_*
  }{
    \displaystyle\int_{-\infty}^{r_*^m}\abs{\psi}^2\dd r_*
  }
  }.
  \label{eq:growthidentity}
\end{equation}

If $e=0$, $V_{\rm eff}\geq0$, the future-horizon condition is imposed, and
the outer boundary satisfies the Dirichlet condition \eqref{eq:mirrorBC}, there is no nontrivial mode with
$\omega_I>0$.
For $e=0$, the first term on the right-hand side of
\eqref{eq:growthidentity} vanishes.  The remaining Rayleigh quotient is
nonnegative.  Hence $\omega_I^2\leq0$, contradicting $\omega_I>0$ unless
$\psi$ is trivial.
A sufficiently negative $V_{\rm eff}$ can support a separate tachyonic
instability.

\subsection{Necessary interval and growth bound in the charged string model}

For \eqref{eq:Atx}, $A_t$ decreases monotonically from
$A_H=a$ to
\begin{equation}
  A_m=a\ee^{-\lambda x_m}
  \label{eq:Am}
\end{equation}
at the mirror.  If $a>0$ and $e>0$, the weighted-average identity
\eqref{eq:omegarealaverage} implies
\begin{equation}
  \boxed{
  ea\ee^{-\lambda x_m}<\omega_R<ea
  }.
  \label{eq:cavityinterval}
\end{equation}
The inequalities are strict for a nontrivial mode with support throughout
the open interval.  A growing cavity mode thus satisfies the horizon condition $\omega_R<ea$. The propagation threshold $\omega>M_\infty$ in \eqref{eq:gapwindow} applies to waves incident from infinity.

For any real quantity restricted to
$A_m\leq A_t\leq A_H$, Popoviciu's inequality \cite{Popoviciu1935} gives
\begin{equation}
  \Var(A_t)\leq\frac{(A_H-A_m)^2}{4}.
  \label{eq:variancebound}
\end{equation}
For the weight in \eqref{eq:weightedaverage}, set $c=(A_H+A_m)/2$. Then $\Var(A_t)=\avg{(A_t-c)^2}-(\avg{A_t}-c)^2\leq (A_H-A_m)^2/4$ since $|A_t-c|\leq(A_H-A_m)/2$.
If $V_{\rm eff}\geq0$, the Rayleigh quotient in
\eqref{eq:growthidentity} is also nonnegative.  We therefore obtain the
parameter-level upper bound
\begin{equation}
  \boxed{
  0<\omega_I
  \leq
  \frac{\abs e}{2}
  \left(A_H-A_m\right)
  =
  \frac{\abs e\,a}{2}
  \left(1-\ee^{-\lambda x_m}\right)
  }.
  \label{eq:growthbound}
\end{equation}
The bound drops the gradient and potential terms in \eqref{eq:growthidentity}
and uses only the range of $A_t$. We compare it with the mode frequencies below.

\subsection{Mirror eigenfrequencies and the onset of growth}
\label{sec:cavity-spectrum}

We solved the Dirichlet problem for $\mu=\xi=0$ and
$h=\exp(\eta X)$ with $\eta=0$ and $1$. Write
$X=\lambda x$, $w=\omega/\lambda$, and $q=ea/\lambda$.
Extracting the ingoing phase,
$\psi=\exp[-\ii\int^X(w-q\ee^{-Y})\,\dd Y/F(Y)]u(X)$,
gives a regular branch with $u(0)=1$. In the following equation primes
denote $X$ derivatives:
\begin{equation}
 F u''+[F'-2\ii(w-q\ee^{-X})]u'
 -[\ii q\ee^{-X}+\eta F'/2+\eta^2F/4]u=0.
 \label{eq:cavityshoot}
\end{equation}
We initialized the analytic solution through quadratic order at
$X_0=10^{-5}$, integrated to $X_m$ with DOP853 (relative and
absolute tolerances $10^{-9}$ and $10^{-11}$), and solved $u(X_m;w)=0$
by a complex secant iteration. A real marginal mode has zero flux
at the mirror. Conservation of flux at the horizon then fixes $w=q$.
At this frequency the regular solution can be chosen real in the
original radial equation. Its first three positive zeros give the
critical mirror radii $X_c$ for the fundamental and two overtones. At onset branch $n$ has $n$ interior nodes. We tracked each complex
root from its own marginal radius as the mirror moved.

\begin{table}[htbp]
\centering
\caption{Critical Dirichlet mirror radii and complex frequencies
$w=\omega/\lambda$ at two common mirror positions. The probe is massless,
$\eta=0$ denotes $h=1$, and $\eta=1$ denotes the dilaton weight.
All frequencies are in the gauge $A_t(\infty)=0$.}
\label{tab:cavity-modes}
\small
\begin{tabular}{@{}rrrrrrr@{}}
\toprule
$a$ & $e/\lambda$ & $\eta$ & $n$ & $X_c$ & $w(X_m=14)$ & $w(X_m=18)$\\
\midrule
0.7 & 1 & 0 & 0 & 2.47703 & $0.20410+0.02164\ii$ & $0.16233+0.01605\ii$\\
0.7 & 1 & 0 & 1 & 6.89441 & $0.40779+0.01859\ii$ & $0.32395+0.01661\ii$\\
0.7 & 1 & 0 & 2 & 11.38157 & $0.60002+0.01083\ii$ & $0.48392+0.01285\ii$\\
\addlinespace[2pt]
0.7 & 1 & 1 & 0 & 5.22392 & $0.54579+0.00075\ii$ & $0.52870+0.00038\ii$\\
0.7 & 1 & 1 & 1 & 11.63403 & $0.65417+0.00132\ii$ & $0.60322+0.00109\ii$\\
0.7 & 1 & 1 & 2 & 18.04677 & $0.77905-0.00944\ii$ & $0.70074-0.00003\ii$\\
\addlinespace[2pt]
0.8 & 1 & 0 & 0 & 1.94791 & $0.20668+0.02113\ii$ & $0.16384+0.01567\ii$\\
0.8 & 1 & 0 & 1 & 5.73984 & $0.41424+0.01802\ii$ & $0.32770+0.01605\ii$\\
0.8 & 1 & 0 & 2 & 9.66393 & $0.61296+0.01216\ii$ & $0.49060+0.01256\ii$\\
\addlinespace[2pt]
0.8 & 1 & 1 & 0 & 3.53875 & $0.54710+0.00080\ii$ & $0.52933+0.00040\ii$\\
0.8 & 1 & 1 & 1 & 8.55068 & $0.65956+0.00215\ii$ & $0.60575+0.00126\ii$\\
0.8 & 1 & 1 & 2 & 13.58113 & $0.78971+0.00077\ii$ & $0.70715+0.00195\ii$\\
\addlinespace[2pt]
0.8 & 1.25 & 0 & 0 & 1.48251 & $0.21349+0.02001\ii$ & $0.16787+0.01459\ii$\\
0.8 & 1.25 & 0 & 1 & 4.42814 & $0.42878+0.01850\ii$ & $0.33664+0.01571\ii$\\
0.8 & 1.25 & 0 & 2 & 7.56010 & $0.63708+0.01431\ii$ & $0.50452+0.01341\ii$\\
\addlinespace[2pt]
0.8 & 1.25 & 1 & 0 & 2.23933 & $0.55011+0.00088\ii$ & $0.53077+0.00043\ii$\\
0.8 & 1.25 & 1 & 1 & 5.79820 & $0.67103+0.00244\ii$ & $0.61128+0.00138\ii$\\
0.8 & 1.25 & 1 & 2 & 9.42391 & $0.81945+0.00356\ii$ & $0.72010+0.00230\ii$\\
\bottomrule
\end{tabular}
\end{table}

Each of the 18 branches has $\omega_I<0$ at $X_c-0.25$ and
$\omega_I>0$ at $X_c+0.25$. At $(a,e/\lambda)=(0.8,1)$, the dilaton
weight moves the first onset from $X_c=1.94791$ to $3.53875$ and
reduces the fundamental growth rate at $X_m=14$ from
$0.021129\lambda$ to $0.000799\lambda$.
At $(a,e/\lambda)=(0.7,1)$
with $\eta=1$, the fundamental and first overtone grow at $X_m=14$,
while the second overtone remains damped even at $X_m=18$.

The supplementary tables give frequencies at $X_m=14,18$ and at
$X_m-X_c=-0.25,0.25,1$ for each branch. All 70 growing modes among
these 90 roots satisfy \eqref{eq:cavityinterval}. Their growth rates
range from approximately $0.044\%$ to $7.61\%$ of the upper bound
\eqref{eq:growthbound}. For example, at $a=0.8$, $e/\lambda=1$,
and $X_m=14$, the allowed interval is
$6.65\times10^{-7}<\omega_R/\lambda<0.8$ (lower end rounded),
and the rate bound is $0.40000\lambda$ for either weight.
The bound neglects the positive gradient contribution and, for
$\eta=1$, the positive effective potential as well. Reducing $X_0$ to
$10^{-6}$ and both integration tolerances by a factor of 100 changed
the complex frequencies by less than $10^{-10}$ in $w$.
An independent 85-point Chebyshev collocation of \eqref{eq:cavityshoot},
with the regular horizon equation imposed at $X=0$, reproduced all
36 frequencies at the two common mirror positions to within
$10^{-10}$ in complex modulus.
As a check for the dilaton weight, zeros of the hypergeometric solution
\eqref{eq:dilatonhypergeom} reproduce all 45 complex frequencies and
all nine marginal mirror radii for this weight, with absolute discrepancies
below $6\times10^{-11}$ in $w$ and $4\times10^{-11}$ in $X_c$, respectively.
\section{Scattering coefficients and amplification}
\label{sec:numerics}

We perform a direct real-frequency integration for the two probe weights. The intrinsic
massless probe has
\begin{equation}
  h=1,\qquad \mu=\xi=0,\qquad V_{\rm eff}=0
  \label{eq:numericalmodel}
\end{equation}
on the background \eqref{eq:Fx}--\eqref{eq:Atx}.  Define
\begin{equation}
  p=\frac{\dd\psi}{\dd r_*}.
  \label{eq:pdef}
\end{equation}
Since $\dd/\dd r_*=F\,\dd/\dd x$, equation
\eqref{eq:schrodinger} is equivalent to the first-order system
\begin{align}
  \frac{\dd\psi}{\dd x}
  &=\frac{p}{F(x)},
  \label{eq:num1}\\
  \frac{\dd p}{\dd x}
  &=-\frac{[\omega-eA_t(x)]^2{}-V_{\rm eff}(x)}{F(x)}\psi.
  \label{eq:num2}
\end{align}
For Table~\ref{tab:numerics}, at $x=x_0\ll1$, the ingoing normalization is chosen as
\begin{equation}
  \psi(x_0)=1,\qquad
  p(x_0)=-\ii(\omega-ea).
  \label{eq:numinitial}
\end{equation}
At a sufficiently large $x=x_{\max}$, $F\simeq1$ and $A_t\simeq0$. For either weight, $k$ is given by \eqref{eq:kdef}.
Writing
\begin{equation}
  \psi
  =\mathcal I\ee^{-\ii k r_*}
  +\mathcal R\ee^{+\ii k r_*},
  \qquad
  p
  =-\ii k\mathcal I\ee^{-\ii k r_*}
  +\ii k\mathcal R\ee^{+\ii k r_*},
  \label{eq:nummatching}
\end{equation}
we solve algebraically for the phase-dressed amplitudes:
\begin{align}
  \mathcal I\ee^{-\ii k r_*}
  &=\frac{1}{2}\left(\psi+\frac{\ii p}{k}\right),
  \label{eq:incidentextract}\\
  \mathcal R\ee^{+\ii k r_*}
  &=\frac{1}{2}\left(\psi-\frac{\ii p}{k}\right).
  \label{eq:reflectedextract}
\end{align}
Because the horizon amplitude in \eqref{eq:numinitial} is unity, the
amplitude ratios used in \eqref{eq:fluxidentity} are the reflectance and the squared horizon amplitude ratio,
\begin{equation}
  \abs{\mathcal R_{\rm ratio}}^2
  =
  \frac{
    \abs{\psi-\ii p/k}^2
  }{
    \abs{\psi+\ii p/k}^2
  },
  \qquad
  \abs{\mathcal T_{\rm ratio}}^2
  =
  \frac{4}{
    \abs{\psi+\ii p/k}^2
  }.
  \label{eq:ratios}
\end{equation}

Table~\ref{tab:numerics} uses
\begin{equation}
  \lambda=1,\qquad a=0.8,\qquad e=1,
  \qquad x_0=10^{-7},\qquad x_{\max}=32.
  \label{eq:numparameters}
\end{equation}
The corresponding temperature is
$T_H=(1-0.8^2)/(4\pi)=0.0286479$, and the predicted threshold is
$\omega_{\rm SR}=ea=0.8$.  We integrate
\eqref{eq:num1}--\eqref{eq:num2} with an eighth-order adaptive
Dormand--Prince method, relative tolerance $3\times10^{-11}$, absolute
tolerance $3\times10^{-13}$, and maximum step $0.03$.

\begin{table}[htbp]
\centering
\caption{Direct scattering check for the charged two-dimensional string
black hole.  The column $\abs{\mathcal T}^2$ gives the squared horizon amplitude ratio. The Wronskian residual is
$\Delta_W=\abs{\mathcal R}^2-
[1-(\omega-ea)\abs{\mathcal T}^2/k]$. Here $h=1$ and $k=\omega$.}
\label{tab:numerics}
\begin{tabular}{@{}rrrrr@{}}
\toprule
$\omega$ &
$\abs{\mathcal R}^2$ &
$\abs{\mathcal T}^2$ &
$Z=\abs{\mathcal R}^2-1$ &
$\abs{\Delta_W}$\\
\midrule
0.15 & 3.2235429522 & 0.5131252967 & 2.2235429522 & $5.0\times10^{-11}$\\
0.25 & 3.6713107257 & 1.2142321480 & 2.6713107257 & $1.6\times10^{-10}$\\
0.40 & 3.0705032908 & 2.0705032907 & 2.0705032908 & $1.7\times10^{-10}$\\
0.70 & 1.9437549616 & 6.6062847313 & 0.9437549616 & $1.4\times10^{-11}$\\
0.79 & 1.1491320912 & 11.7814352065 & 0.1491320912 & $4.6\times10^{-14}$\\
0.90 & 0.0810455711 & 8.2705898603 & $-0.9189544289$ & $1.4\times10^{-11}$\\
\bottomrule
\end{tabular}
\end{table}

Every sampled frequency below $0.8$ is amplified, while the point above
$0.8$ is absorbed.  The maximum tabulated Wronskian residual is
$1.7\times10^{-10}$.  At fixed cutoffs, tightening both integration tolerances by a factor of 100 reduces the maximum residual to $2.36\times10^{-12}$ and changes the six gains by at most $1.29\times10^{-12}$. These entries are finite-cutoff values. Changing
$\lambda x_0$ from $10^{-7}$ to $10^{-6}$ or $10^{-8}$ changes the six
tabulated gains by at most $6.0\times10^{-6}$ or $3.6\times10^{-7}$,
respectively. The Wronskian residual measures flux balance at a given
cutoff. The outer-cutoff test is reported below.

\subsection{Scattering across the open band}
\label{sec:scatteringcurves}

Figure~\ref{fig:scatteringwindows} compares the two massless probes at four
parameter choices. For these curves we integrate in $\ln(\lambda x)$
from $\lambda x_0=10^{-9}$ to $\lambda x_{\max}=32$ and match with the
wavenumber $k$ of \eqref{eq:kdef}. At $a=0.8$ and $e=\lambda=1$, the
dilaton-weighted field gives $\abs{\mathcal R/\mathcal I}^2=
1.18167,\ 1.25618,\ 1.07806$ and $0.44430$ at
$\omega=0.60,\ 0.70,\ 0.79$ and $0.85$, respectively. The peak moves
from $\omega/\lambda=0.23495$ for $h=1$ to $0.71617$ for
$h\propto \ee^{\lambda x}$, and the corresponding gain falls from
$268.16\%$ to $25.92\%$. Figure~\ref{fig:scatteringpeaks} tracks both
peaks as $a$ varies at fixed $e/\lambda=1$. The dilaton branch opens at
$a=1/2$. Its peak grows as the window widens over the sampled range.

\begin{figure}[htbp]
\centering
\includegraphics[width=0.96\linewidth]{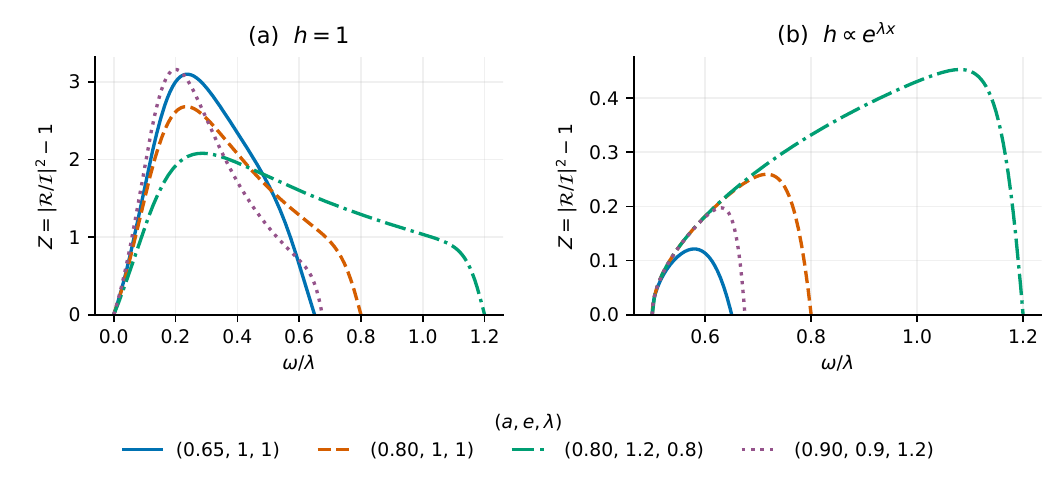}
\caption{Massless amplification spectra over each open band. The labels
give $(a,e,\lambda)$, with $e$ and $\lambda$ in a common inverse-length
unit. Each curve includes 149 interior frequencies. The endpoints show
the limiting gain. The dilaton gap fixes the lower edge of the right panel
at $\omega/\lambda=1/2$.}
\label{fig:scatteringwindows}
\end{figure}

\begin{figure}[htbp]
\centering
\includegraphics[width=0.96\linewidth]{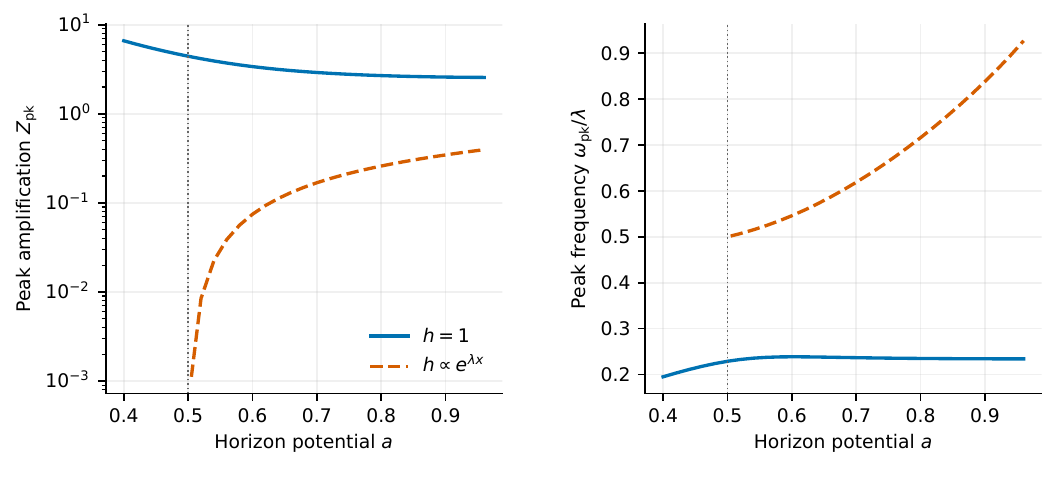}
\caption{Peak amplification (logarithmic ordinate) and its frequency at
$e=\lambda$, $\mu=\xi=0$. We bracket each resolved maximum on a
129-point frequency grid and refine it by one-dimensional maximization.
The dotted line marks the onset $a=1/2$ of the dilaton-weighted band.
Temperature follows from $4\pi T_H/\lambda=1-a^2$.}
\label{fig:scatteringpeaks}
\end{figure}

We tested the outer matching radius at the four frequencies quoted above,
for both weights at $a=0.8$ and $e/\lambda=1$. Relative to
$\lambda x_{\max}=40$, the largest absolute changes in $Z$ were
$1.14\times10^{-5}$, $6.43\times10^{-8}$ and $3.56\times10^{-11}$
at $\lambda x_{\max}=12,16,24$, respectively.
A separate long-wavelength check is available for $h=1$. Let
$s=\omega/\lambda$ and $g=ea/\lambda$, and take $g\to0$ at fixed
nonextremal $a$ and fixed $0<s/g<1$, away from $s/g=1/2$.
Both wavelengths then exceed the transition scale of $A_t$ in the
tortoise coordinate. At leading order $\psi$ and
$\dd\psi/\dd r_*$ match across the transition as across a step, with
signed wavenumbers $\lambda s$ and $\lambda(s-g)$. This gives
$\mathcal R/\mathcal I\to g/(2s-g)$ up to a phase.
For $s/g=1/4$, direct integration gives
$\abs{\mathcal R/\mathcal I}^2=4.02087$, $4.00532$ and $4.00134$
at $g=0.04,0.02,0.01$, approaching the step value of four.

\subsection{Analytic check for the dilaton-weighted field}
\label{sec:analyticreflectance}

For $h\propto\ee^{\lambda x}$ and $\xi=0$, the scattering coefficients
can be obtained from a hypergeometric connection formula. Related
reflection amplitudes for the charged coset background appear in
Refs.~\cite{Giveon2003,Gursoy2022}. In the present probe conventions,
put $\rho=\ee^{\lambda x}$, $b=a^2$, $D=(\rho-1)(\rho-b)$,
and $K=k/\lambda=\sqrt{s^2-(\mu/\lambda)^2-1/4}$. The original
radial field $R$ of \eqref{eq:separation} obeys
\begin{equation}
 \frac{\dd}{\dd\rho}\!\left(D\frac{\dd R}{\dd\rho}\right)
 +\left[\frac{(s\rho-g)^2}{D}-\frac{\mu^2}{\lambda^2}\right]R=0 .
 \label{eq:dilatonexactode}
\end{equation}
Set $z=(\rho-1)/(\rho-b)$, $\alpha=-\ii(s-g)/(1-b)$,
$\beta=1/2+\ii K$ and $\delta=(bs-g)/(1-b)$. The ingoing solution is
\begin{equation}
 R=z^\alpha(1-z)^\beta\,
 {}_2F_1(\alpha+\beta+\ii\delta,\alpha+\beta-\ii\delta;
                         1+2\alpha;z).
 \label{eq:dilatonhypergeom}
\end{equation}
At $z=1$, the powers $(1-z)^{1/2+\ii K}$ and $(1-z)^{1/2-\ii K}$
are the incoming and outgoing waves after multiplication by $\sqrt h$.
Using the connection coefficients and
$\abs{\Gamma(1/2+\ii t)}^2=\pi/\cosh(\pi t)$ gives, for $K>0$,
\begin{equation}
 \left|\frac{\mathcal R}{\mathcal I}\right|^2
 =\frac{\cosh[\pi(s-K)]\cosh[\pi(d+K)]}
        {\cosh[\pi(s+K)]\cosh[\pi(d-K)]},
 \qquad d=\frac{2g-(1+b)s}{1-b}.
 \label{eq:dilatonreflectance}
\end{equation}
The largest absolute difference in $Z$ between this expression and the
596 interior points in the right panel of Fig.~\ref{fig:scatteringwindows}
is $4.4\times10^{-9}$ or less.

Since $\cosh[\pi(d+K)]/\cosh[\pi(d-K)]<\exp(2\pi K)$, the
dilaton-weighted gain satisfies $Z<1$ throughout the open superradiant
band for $\mu^2\geq0$. At fixed $s$ strictly inside the limiting band,
$K>0$ and $s<e/\lambda$, we have $d\to+\infty$ as $a\to1^-$.
Equation~\eqref{eq:dilatonreflectance} then gives the finite limit
\begin{equation}
 \lim_{a\to1^-}\left|\frac{\mathcal R}{\mathcal I}\right|^2
 =\frac{1+\exp[-2\pi(s-K)]}{1+\exp[-2\pi(s+K)]}.
 \label{eq:dilatonextremal}
\end{equation}
For the constant-weight field, the numerical gains at $e=\lambda$ and
$s=0.40$ are $1.99746$, $1.97759$ and $1.97741$ at $a=0.9$, $0.99$
and $0.999$, respectively. These fixed-frequency results do not control
a joint limit in which the frequency also approaches the horizon edge.
\section{Consistency checks and limitations}
\label{sec:tests}

The preceding derivations give the following consistency checks.
\begin{enumerate}[label=(\arabic*),leftmargin=2.2em]
  \item A real-frequency calculation satisfying the assumptions of
  Section~\ref{sec:flux} must obey \eqref{eq:fluxidentity}.  A violation is
  evidence of inconsistent asymptotic normalization, a sign error, or
  insufficient numerical accuracy.

  \item A neutral static problem cannot have
  $\abs{\mathcal R}>1$ merely because $V_{\rm eff}<0$ somewhere.  Such a
  claim would contradict \eqref{eq:neutralnoSR}.

  \item A growing charged cavity mode in the string background with $a>0$ and $e>0$ must satisfy
  \eqref{eq:cavityinterval}. When $V_{\rm eff}\geq0$ it also satisfies
  \eqref{eq:growthbound}.

  \item In the limits $e\to0$ or $a\to0$, the charged amplification window
  must close continuously.
\end{enumerate}

The calculation neglects backreaction of the scalar on the metric,
dilaton, and gauge field. The cavity survey follows three mode branches
for selected charges and mirror positions. The exactly extremal horizon
has a different local wave equation, and the nonlinear endpoint of a
growing mode lies outside linear perturbation theory.

\section{Conclusion}

We have reconstructed the scalar perturbation problem for static
two-dimensional black holes from the action to the scattering and spectral
conditions.  The exact field redefinition shows which terms genuinely
contribute to the effective potential and why an angular barrier cannot be
borrowed from a higher-dimensional spacetime without an explicit
dimensional reduction.

The horizon potential fixes the upper edge of the superradiant band for
both kinetic weights. For the string-frame scalar, the dilaton gradient
adds $\lambda^2/4$ to the asymptotic mass squared. For $e>M_\infty$ at fixed $\mu$ and $\lambda$, the two edges meet at $T_c$ in \eqref{eq:Tcritical} and the band closes at a nonzero horizon potential. For $0<e\leq M_\infty$ the band is absent throughout the nonextremal family.
In the massless benchmark, the peak gain shifts to a higher frequency
and falls by about a factor of ten when the string-frame weight is included.

For the Dirichlet wall, the complex-frequency integral identity
separates the two ingredients needed for a genuine instability.  The gauge
potential supplies the positive variance term, whereas gradients and a
nonnegative effective potential supply a stabilizing Rayleigh quotient.
For the computed mirror modes the integral bound exceeds the measured
growth rates. Tracking each frequency through its real-axis crossing
locates the onset of instability, including the displacement of the
fundamental onset from $\lambda x_m=1.948$ to $3.539$ in the benchmark.

\section*{Acknowledgments}

This work was supported by the National Natural Science Foundation of China
(General Program) under Project Approval No.~11574089.

\section*{Author contributions}

Wen-Xiang Chen conceived the study, performed the analytical derivations and
numerical checks, and wrote the manuscript.

\section*{Data availability}

No external data were used.  All numerical parameters, evolution equations,
boundary conditions, and matching formulas required to reproduce
Table~\ref{tab:numerics} are provided in Section~\ref{sec:numerics}.

\section*{Conflict of interest}

The author declares no conflict of interest.

\appendix

\section{Noether current and its relation to the Wronskian}
\label{app:current}

The global phase symmetry $\Psi\mapsto\ee^{\ii\alpha}\Psi$ of
\eqref{eq:action} gives the conserved current
\begin{equation}
  j^\mu
  =\frac{h}{2\ii}
  \left[
    \Psi^*\cD^\mu\Psi
    -\Psi(\cD^\mu\Psi)^*
  \right],
  \qquad
  \nabla_\mu j^\mu=0.
  \label{eq:current}
\end{equation}
For the separated field at real frequency and radial metric,
\begin{align}
  j^x
  &=
  \frac{hF}{2\ii}
  \left(R^*\frac{\dd R}{\dd x}
  -R\frac{\dd R^*}{\dd x}\right)
  \nonumber\\
  &=
  \frac{h}{2\ii}
  \left(R^*\frac{\dd R}{\dd r_*}
  -R\frac{\dd R^*}{\dd r_*}\right).
  \label{eq:jxR}
\end{align}
Using $R=\psi/\sqrt h$ and the reality of $h$, all derivatives of $h$
cancel in the antisymmetric combination:
\begin{equation}
  j^x
  =\frac{1}{2\ii}
  \left(\psi^*\psi'-\psi\psi'^*\right)
  =\frac{W}{2\ii}.
  \label{eq:jxW}
\end{equation}
Thus Wronskian conservation in Section~\ref{sec:flux} is the radial Noether
flux conservation law.

\section{Where an angular potential actually comes from}
\label{app:reduction}

Consider a spherically symmetric $D$-dimensional spacetime,
whose two-dimensional orbit-space metric has the form
\begin{equation}
  \dd s_2^2=-F(r)\dd t^2+\frac{\dd r^2}{F(r)}.
  \label{eq:Dmetric}
\end{equation}
After expanding a $D$-dimensional scalar in harmonics on
$S^{D-2}$, the radial action acquires
\begin{equation}
  h(r)=r^{D-2},
  \qquad
  U(r)=\mu^2+\frac{\ell(\ell+D-3)}{r^2}.
  \label{eq:Dweight}
\end{equation}
Let $n=D-2$.  Since $\dd/\dd r_*=F\,\dd/\dd r$,
\begin{equation}
  \frac{\dd\ln h}{\dd r_*}
  =\frac{nF}{r}.
  \label{eq:HdotD}
\end{equation}
Differentiating once more,
\begin{equation}
  \frac{\dd^2\ln h}{\dd r_*^2}
  =
  nF\frac{\dd}{\dd r}\left(\frac{F}{r}\right)
  =
  nF\left(\frac{F'}{r}-\frac{F}{r^2}\right).
  \label{eq:HddotD}
\end{equation}
Substitution in \eqref{eq:Veff} yields
\begin{equation}
  V_{\rm eff}^{(D)}
  =
  F\left[
    \mu^2+\frac{\ell(\ell+D-3)}{r^2}
  \right]
  +\frac{nFF'}{2r}
  +\frac{n(n-2)F^2}{4r^2}.
  \label{eq:VD}
\end{equation}
For $D=4$, $n=2$, and
\begin{equation}
  V_{\rm eff}^{(4)}
  =
  F\left[
    \mu^2+\frac{\ell(\ell+1)}{r^2}
    +\frac{F'}{r}
  \right],
  \label{eq:V4D}
\end{equation}
which is the standard scalar potential.  Equation \eqref{eq:V4D} is a
result of the spherical area weight and harmonic eigenvalue.  Neither
ingredient exists in an intrinsic $(1+1)$-dimensional model with constant
$h$.

\end{document}